\documentclass[sigconf,natbib=true,anonymous=false]{acmart}
\usepackage{multirow}
\usepackage{enumitem}
\usepackage{arydshln}
\usepackage{graphicx} 		
\usepackage{subfigure}		
\usepackage{amsmath}
\usepackage{algorithm}
\usepackage{bm}
\usepackage{algorithmic} 

\newcommand{\ourname}{$\text{T$^2$ARec}$}
\newcommand{\ourblockname}{$\text{T$^2$A-Mamba}$}
\newcommand{\ourssmname}{$\text{Align}^2\text{-SSM}$}
\newcommand{\oursslname}{$\text{interest state}$}

\AtBeginDocument{%
  }

\setcopyright{acmlicensed}
\copyrightyear{2018}
\acmYear{2018}
\acmDOI{XXXXXXX.XXXXXXX}

\acmConference[Conference acronym 'XX]{Make sure to enter the correct
  conference title from your rights confirmation emai}{June 03--05,
  2018}{Woodstock, NY}

\acmISBN{978-1-4503-XXXX-X/18/06}

\begin{document}

\title{A Test-Time Training Approach for Alleviating User Interest Shifts in Sequential Recommendation}
\title{Test-Time Alignment for Tracking User Interest Shifts in Sequential Recommendation}

\author{Changshuo Zhang}
\affiliation{%
  \institution{Gaoling School of AI,\\ Renmin University of
China}
  \city{Beijing}
  \country{China}
}
\email{lyingcs@ruc.edu.cn}

\author{Xiao Zhang}
\authornote{Xiao Zhang is the corresponding author.}
\author{Teng Shi}
\affiliation{%
  \institution{Gaoling School of AI,\\ Renmin University of
China}
    \city{Beijing}
  \country{China}
  }
\email{zhangx89@ruc.edu.cn}
\email{shiteng@ruc.edu.cn}

\author{Jun Xu}
\author{Ji-Rong Wen}
\affiliation{%
  \institution{Gaoling School of AI,\\ Renmin University of
China}
  \city{Beijing}
  \country{China}}
\email{junxu@ruc.edu.cn}
\email{jrwen@ruc.edu.cn}




\renewcommand{\shortauthors}{Changshuo Zhang et al.}

\begin{abstract}
Sequential recommendation is essential in modern recommender systems, aiming to predict the next item a user may interact with based on their historical behaviors. However, real-world scenarios are often dynamic and subject to shifts in user interests. 
Conventional sequential recommendation models are typically trained on static historical data, limiting their ability to adapt to such shifts and resulting in significant performance degradation during testing. 
Recently, Test-Time Training (TTT) has emerged as a promising paradigm, enabling pre-trained models to dynamically adapt to test data by leveraging unlabeled examples during testing. However, applying TTT to effectively track and address user interest shifts in recommender systems remains an open and challenging problem. Key challenges include how to capture temporal information effectively and explicitly identifying shifts in user interests during the testing phase. 
To address these issues, we propose \textbf{\ourname}, a novel model leveraging state space model for TTT by introducing two \textbf{T}est-\textbf{T}ime \textbf{A}lignment modules tailored for sequential recommendation, effectively capturing the distribution shifts in user interest patterns over time. Specifically, \textbf{\ourname} aligns absolute time intervals with model-adaptive learning intervals to capture temporal dynamics and introduce an \oursslname\ alignment mechanism to effectively and explicitly identify the user interest shifts with theoretical guarantees. These two alignment modules enable efficient and incremental updates to model parameters in a self-supervised manner during testing, enhancing predictions for online recommendation. Extensive evaluations on three benchmark datasets demonstrate that \ourname\ achieves state-of-the-art performance and robustly mitigates the challenges posed by user interest shifts. 
\end{abstract}



\maketitle

\section{Introduction}

\begin{figure}[t]
\centering
\includegraphics[width=0.40 \textwidth]{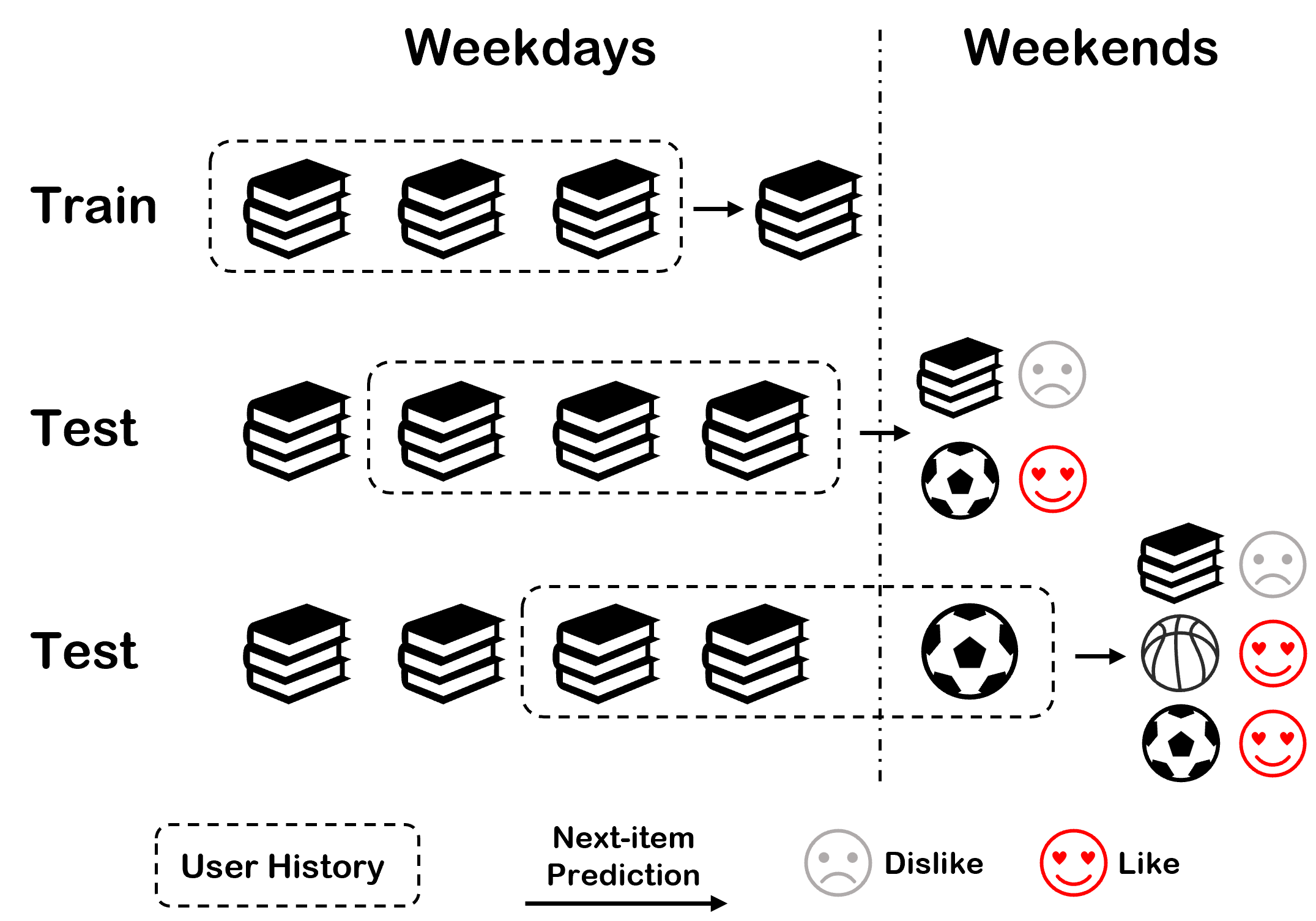}
\caption{Illustrates of user interest shifts between train and testing phases in sequential recommendation. During the training phase (weekdays), the user’s historical behavior is focused on study-related content (e.g., books), and the model learns to predict the next item based on this pattern. In the testing phase (weekends), user interest shifts from study-related items to leisure activities (e.g., sports). Although the Test input includes the ground truth from the training phase, the model continues to recommend study-related content, failing to adapt to the user’s new preferences. This highlights the importance of modeling temporal contexts and handling user interest shifts effectively in recommendations.}
\label{intro}
\end{figure}

Sequential recommendations aim to predict the next item a user will interact with by modeling dependencies within their historical interaction data~\cite{tang2018personalized, hidasi2015session, sun2019bert4rec, kang2018self,liu2024mamba4rec, xie2022contrastive, chang2021sequential, zhang2024reinforcing, fang2020deep, dai2024recode, zhang2024saqrec}. However, real-world scenarios often present the challenge of distribution shifts, where user behavior patterns and data distributions evolve dynamically over time. 
For instance, consider a sequential recommendation where a user’s interaction history (as shown in Figure~\ref{intro}) during the training phase (weekdays) exhibits a strong preference for study-related items, such as books, reflecting their focus on work or learning. However, during the testing phase (weekends), the user’s interest shifts towards sports-related items, such as football or basketball equipment, as they transition to leisure activities. If the recommendation fails to adapt to this shift and continues to recommend study-related items based on the training phase, it will not align with the user’s weekend preferences. 

Existing sequential recommendation models, typically trained on static historical data and have their model parameters fixed during online deployment, face challenges in adapting to shifts in user interest patterns. This limitation often leads to significant performance degradation during test time. 
To validate this phenomenon, we conducted experiments on two datasets (ML-1M~\cite{harper2015movielens} and Amazon Prime Pantry~\cite{ni2019justifying}) using two models (SASRec~\cite{kang2018self} and Mamba4Rec~\cite{liu2024mamba4rec}) follow the implementation settings in Recbole~\cite{zhao2021recbole}. After training the models on the training set, the testing set was evenly divided into four segments based on timestamps, and NDCG@$10$ was used as the evaluation metric. As shown in Figure~\ref{vali}, the later the timestamp of the segment, the more significant the user interest shift, resulting in poorer test performance metrics. This illustrates how user behavior evolves dynamically between the train and testing phases, driven by contextual factors such as time availability, necessitating adaptive models to handle such user interest shifts effectively. 

Motivated by the recent promising paradigm of Test-Time Training (TTT)~\cite{liu2021ttt++, jin2022empowering, wang2022test, sun2020test, liang2024comprehensive}, which enables pre-trained models to dynamically adapt to test data by leveraging unlabeled examples during inference, we focus on applying TTT to track user interest shifts in sequential recommendation. 
While TTT facilitates self-adaptation and effectively addresses evolving distribution shifts in an online manner, its application to tracking user interest shifts in recommender systems remains an open and challenging problem. We identify two key challenges associated with this task: 
first, \emph{capturing temporal information}, as user behavior are often influenced by periodic patterns or trending topics, making it crucial to understand the impact of historical events on predictions at specific future time points; 
second, \emph{dynamically and efficiently adjusting the representation of user interest patterns}, since even if a user’s historical behavioral features remain stable, their interest patterns may evolve dynamically. The model needs to identify and adapt to these changes to provide  accurate recommendations during testing. 




\begin{figure}[t]
\centering
\includegraphics[width=0.45 \textwidth]{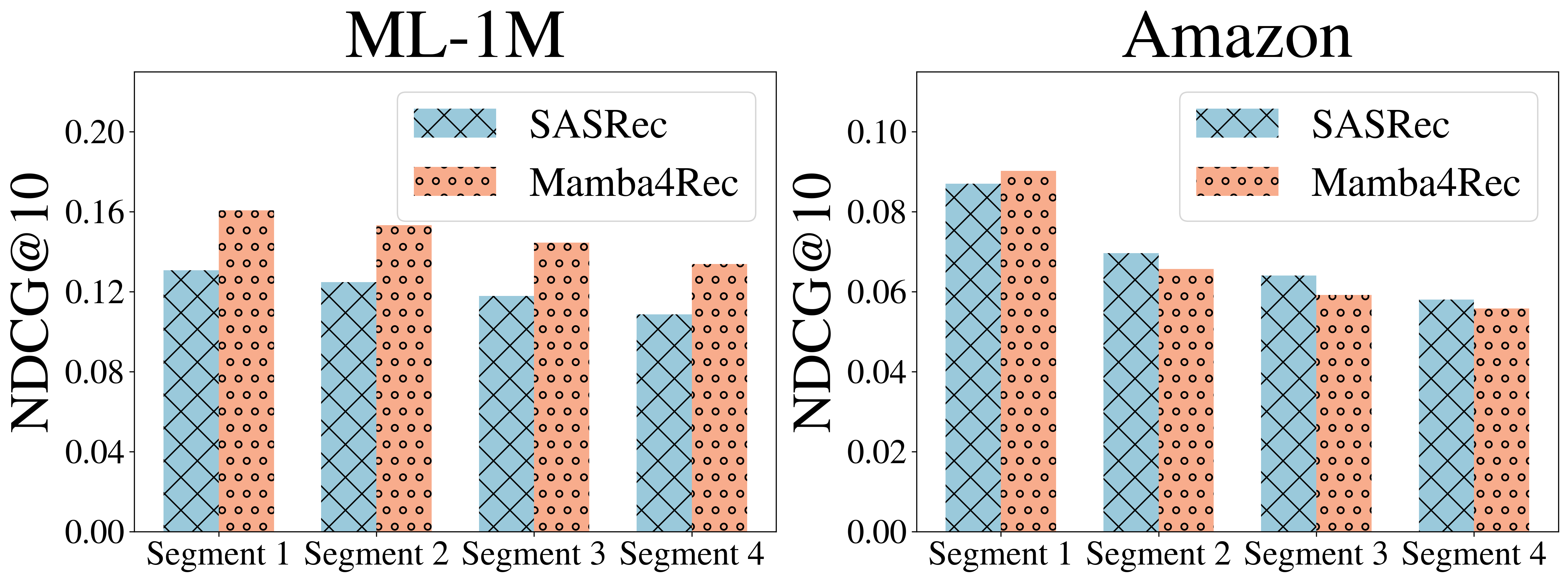}
\caption{Validation of user interest shifts during the testing phase. We conducted experiments on two datasets (ML-1M and Amazon Prime Pantry) using two backbone models (SASRec and Mamba4Rec). After training the models on the training set, the testing set was evenly divided into four segments based on timestamps, and NDCG@10 was used as the evaluation metric for analysis and comparison.}
\label{vali}
\end{figure}



To address the above challenges of tracking user interest shifts in sequential recommendation tasks, we propose \ourname, a sequential model that integrates test-time training. \ourname\ adopts state space models~(SSMs)~\cite{gu2023mamba, dao2024transformers} as the backbone, which effectively models user interest state transitions by handling historical interactions and resolves test-time throughput issues. Then we introduce two alignment-based self-supervised losses to adaptively capture the user interest shifts. The time interval alignment loss computes differences between interaction intervals in the sequence and the target prediction time, aligning these intervals with adaptive time steps to effectively capture temporal dynamics. The \oursslname\ alignment loss models the dynamic evolution of user interest patterns by transforming the input sequence into a final state, applying forward and backward state updates to generate a reconstructed state, and aligning it with the original state for precise pattern modeling. During testing, \ourname\ applies gradient descent on test data to adjust model parameters in real time, enabling accurate next-item predictions under distribution shifts. This design allows \ourname\ to effectively capture temporal dynamics and adapt to evolving user interest patterns, ensuring robust and efficient performance in sequential recommendation tasks at the test time.


Our main contributions are as follows:
\begin{itemize}[leftmargin=*]
    \item Identification of a key issue in applying Test-Time Training to sequential recommendations: The overlooked shift in the user interest pattern at the test time.
    \item Introduction of a novel approach: We propose \ourname, a TTT-based sequential recommendation model incorporating two test-time alignment-based losses in a state space model to capture temporal dynamics and evolving user interest patterns, along with real-time parameter adjustment during testing to track user interest shifts and ensure robust performance.
    \item Extensive experiments: We conduct experiments on three widely used datasets, demonstrating the effectiveness of \ourname. Further ablation studies and analysis explain the superiority of our designed modules.
\end{itemize}

\section{Related Work}
\subsection{Sequential Recommendation}
Sequential recommendations have evolved from traditional models, like Markov Chains~\cite{he2016fusing}, to deep learning approaches. Early RNN-based models (e.g., GRU4Rec~\cite{hidasi2015session}, HRNN~\cite{quadrana2017personalizing}) addressed long-range dependencies, leveraging hidden states to capture dynamic user preferences. Recently, Transformer-based architectures, such as SASRec~\cite{kang2018self} and BERT4Rec~\cite{sun2019bert4rec}, have gained prominence with self-attention mechanisms that model complex interactions and enable efficient parallel computation. Alternative methods, like MLP-based models (e.g., FMLP-Rec~\cite{zhou2022filter}), offer a balance between accuracy and efficiency. Innovations like Mamba4Rec\cite{liu2024mamba4rec} further enhance performance on long interaction sequences, reflecting the ongoing advancements in this field.
\subsection{Test-Time Training}
Test-Time Training (TTT)~\cite{liu2021ttt++, jin2022empowering, wang2022test} improves model generalization by enabling partial adaptation during the testing phase to address distribution shifts between training and testing datasets. It leverages self-supervised learning (SSL)~\cite{liu2022graph, schiappa2023self, yang2024mixed, liu2021self} tasks to optimize a supervised loss (e.g., cross-entropy) and an auxiliary task loss during training.
The auxiliary task (e.g., rotation prediction) allows the model to align test-time features closer to the source domain. Simplified approaches like Tent~\cite{wang2020tent} avoid supervised loss optimization during testing, while advanced methods such as TTT++~\cite{liu2021ttt++} and TTT-MAE~\cite{he2022masked} employ techniques like contrastive learning and masked autoencoding to enhance adaptation. Unsupervised extensions, such as ClusT3~\cite{hakim2023clust3}, use clustering with mutual information maximization but face limitations due to hyperparameter sensitivity. Test-Time Training (TTT) has been applied to out-of-distribution (OOD) tasks, such as recommendations~\cite{yang2024dual, yang2024ttt4rec, wang2024not}. DT3OR introduces a model adaptation mechanism during the test phase, specifically designed to adapt to the shifting user and item features~\cite{yang2024dual}. TTT4Rec leverages SSL in an inner-loop, enabling real-time model adaptation for sequential recommendations~\cite{yang2024ttt4rec}. LAST constructs simulated feedback through an evaluator during testing, updating model parameters and achieving online benefits~\cite{wang2024not}. However, existing work  has not effectively captured temporal information or explicitly identified shifts in user interests during the testing phase.

\section{Preliminaries}
\subsection{Problem Statement}\label{seqrec}
In sequential recommendation, let \( \mathcal{U} = \{ u_1, u_2, \dots, u_{|\mathcal{U}|} \} \) denote the user set, \( \mathcal{V} = \{ v_1, v_2, \dots, v_{|\mathcal{V}|} \} \) denote the item set, and \( \mathcal{S}_u = [ v_1, v_2, \dots, v_{n_u} ] \) denote the chronologically ordered interaction sequence for user \( u \in \mathcal{U} \) with the corresponding timestamp sequence \([ t_1, t_2, \dots, t_{n_u} ] \), where \( n_u \) is the length of the sequence, \( v_i \) is the \( i \)-th item interacted with by user \( u \), and \( t_i \) is the timestamp of the interaction. Given the interaction history \(\mathcal{S}_u \) and the timestamp of prediction \( t_{n_u+1} \), the task is to predict the next interacted item \( v_{n_u+1} \), i.e., the item that the user \( u \) is most likely to interact with at timestamp \( t_{n_u+1} \). This can be formalized as learning a function:
\begin{equation}
f_{t_{n_u+1}} :~~\mathcal{S}_u \rightarrow v_{n_u+1},    
\end{equation}
where \( f_{t_{n_u+1}} \) maps the historical sequence \( \mathcal{S}_u \) at timestamp \( t_{n_u+1} \) to the next likely item \( v_{n_u+1} \) from the item set \( \mathcal{V} \). In the following sections, we omit the subscript \(u\) in \(n_u\) and \(\mathcal{S}_u\) and directly use \(n\) 
 and \(\mathcal{S}\) for convenience. 

\subsection{State Space Models (SSMs)}\label{mamba}
SSMs perform well in long-sequence modeling~\cite{liu2024mamba4rec}, image generation~\cite{xu2024survey}, and reinforcement learning~\cite{ota2024decision}, providing efficient autoregressive inference like RNN~\cite{lipton2015critical} while processing input sequences in parallel like Transformers~\cite{vaswani2017attention}. This dual functionality enables efficient training and robust performance in applications such as time series analysis~\cite{rangapuram2018deep} and audio generation~\cite{goel2022s}.

The original SSMs originated as continuous-time maps on functions from $d$-dimensional input \(\bm x(t) \in \mathbb{R}^{d}\) to output \(\bm y(t) \in \mathbb{R}^{d}\) at current time $t$ through a $d_{\mathrm{s}}$-dimensional hidden state \(\bm h(t) \in \mathbb{R}^{d_{\mathrm{s}}}\). These models leverage the dynamics described below:
\begin{subequations}
\begin{align}
\bm h^{\prime}(t) &= \bm{A} \bm{h}(t) + \bm{B} \bm x(t), \label{ssm1} \\
\bm y(t) &= \bm{C}^\top \bm{h}(t), \label{ssm2}
\end{align}
\end{subequations} 
where \(\bm A \in \mathbb{R}^{d_{\mathrm{s}} \times d_{\mathrm{s}}}\) and \(\bm B,\bm C\in \mathbb{R}^{d_{\mathrm{s}} \times d}\) are adjustable matrices, $\bm h'(t)$ denotes the derivative of $\bm h(t)$. To enable effective representation of discrete data, Structured SSMs~\cite{gu2021efficiently} employ the Zero-Order Hold (ZOH)~\cite{gu2021combining} method for data discretization from the input sequence $\bm X = [\bm x_1, \bm x_2,\ldots,\bm x_n]^\top\in \mathbb{R}^{n\times d}$ to output sequence $\bm Y = [\bm y_1, \bm y_2,\ldots,\bm y_n]^\top \in \mathbb{R}^{n\times d}$ through hidden state $\bm H = [\bm h_1, \bm h_2,\ldots,\bm h_n]^\top  \in \mathbb{R}^{n\times d_{\mathrm{s}}}$ based on a specified step size $\Delta \in \mathbb{R}$.
The introduction of the Mamba~\cite{gu2023mamba} model significantly enhances SSMs by dynamically adjusting the matrices \(\bm B\in \mathbb{R}^{n\times d_{\mathrm{s}}}\), \(\bm C\in \mathbb{R}^{n\times d_{\mathrm{s}}}\) and the step size now represented as \(\bm\Delta\in \mathbb{R}^{n\times d}\) that dynamically depends on inputs varying over time. 
Its latest version, Mamba-2~\cite{dao2024transformers}, links structured state space models with attention mechanisms. 
Mamba-2 refines \(\bm A\) into a scalar value $A\in \mathbb{R}$ and set $\bm \Delta \in \mathbb{R}^{n}$, creating a new state space duality (SSD) framework with multi-head patterns akin to Transformers. 
This innovation boosts training speed and increases state size, optimizing expressiveness for greater efficiency and scalability, making it a strong contender against Transformers.

\begin{figure}[t]
\centering
\includegraphics[width=0.4 \textwidth]{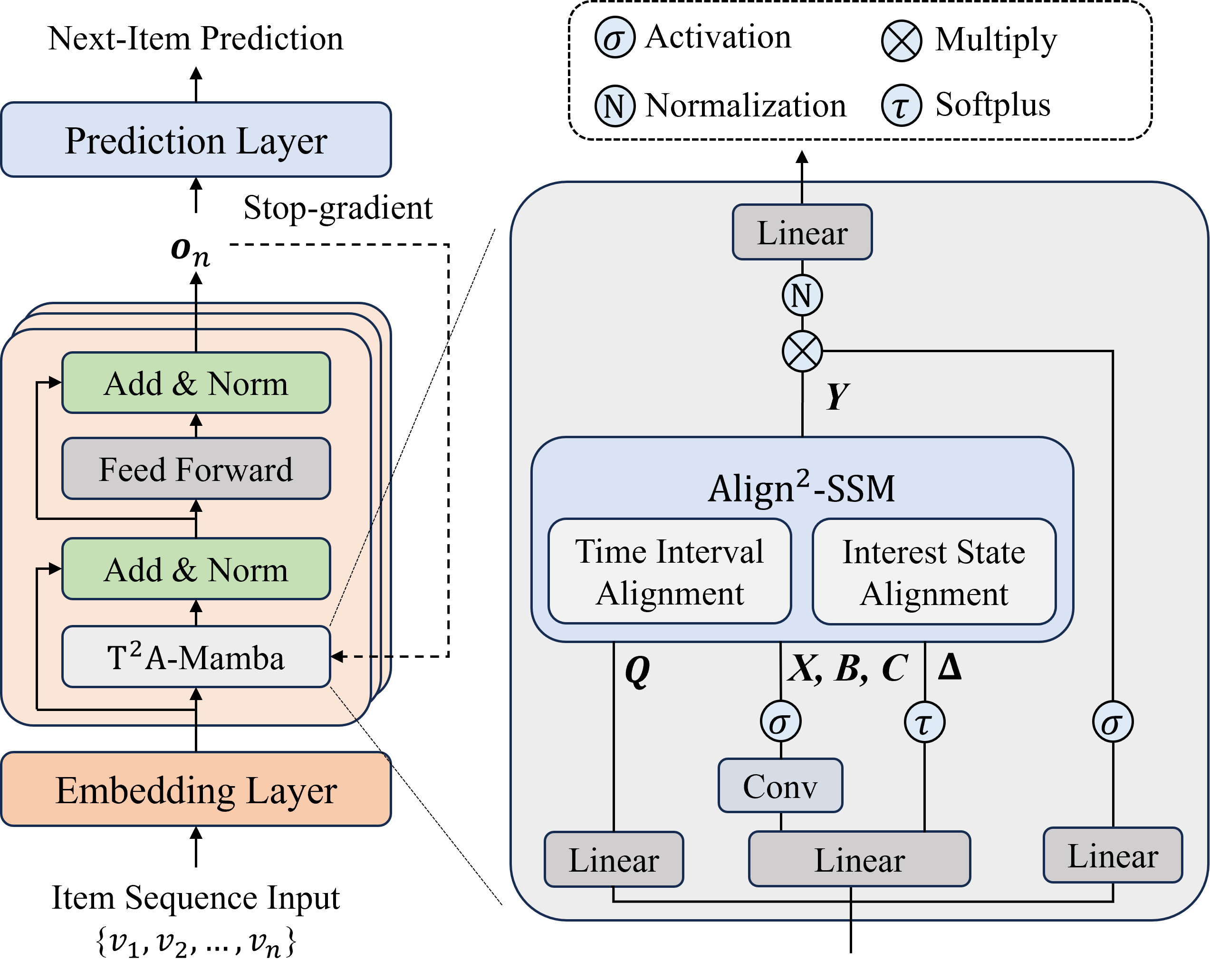}
\caption{Overall framework of \ourname: \ourname\ processes input sequences through an embedding layer, followed by the \ourblockname\ block and \ourssmname\ block for state updates and output generation. The prediction layer using output embedding $\bm{o}_n$ generated from the feed forward network layer for next-item predictions. $\bm{o}_n$ is reintroduced into the \ourblockname\ block to compute the alignment losses.}
\label{main_model1}
\end{figure}

\begin{figure*}[t]
\centering
\includegraphics[width=0.9 \textwidth]{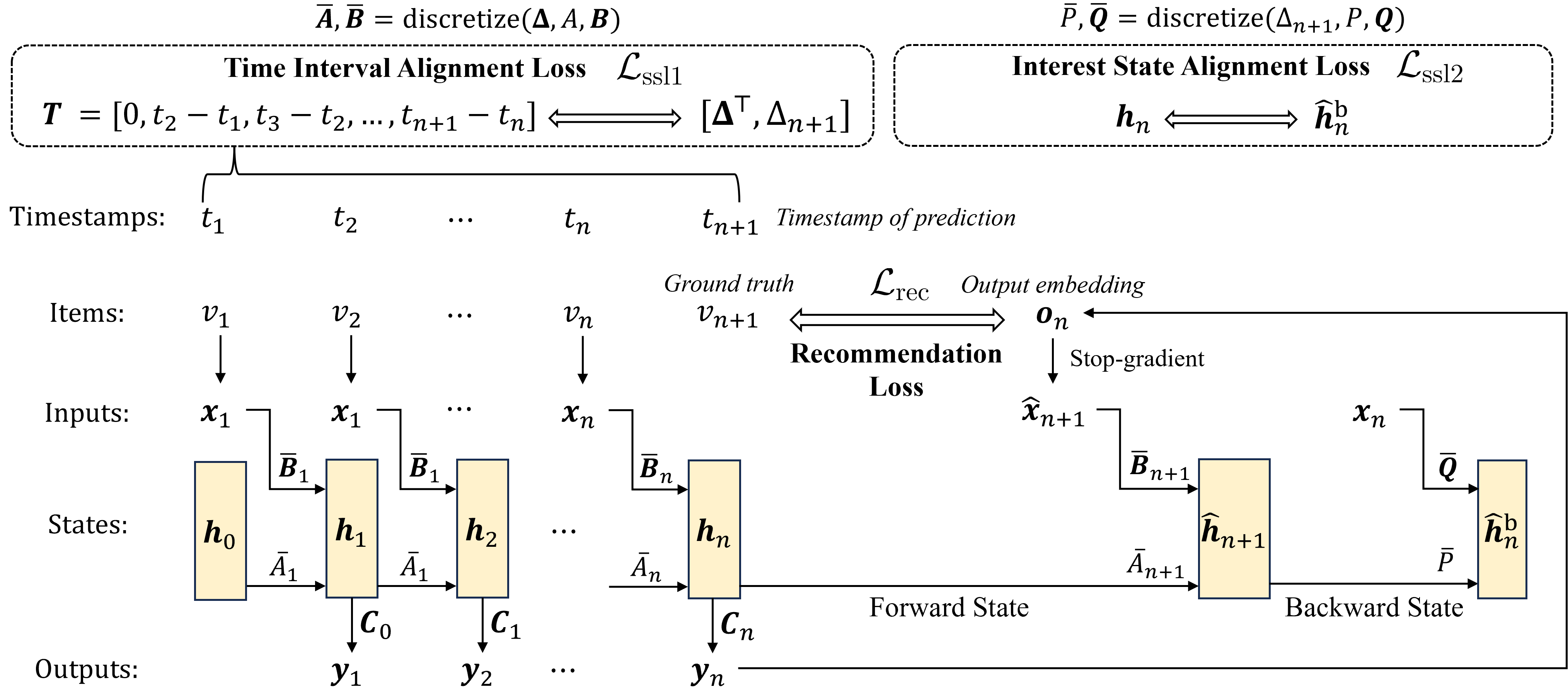}
\caption{The logits and losses computation in \ourssmname: The left side illustrates the time interval alignment loss (\(\mathcal{L}_{\mathrm{time}}\)), which aligns the predicted time intervals \(\bm\Delta\) with the ground truth \(\bm T\). The right side shows the \oursslname\ alignment loss (\(\mathcal{L}_{\mathrm{state}}\)), aligning the final state \(\bm h_{n}\) with the backward state \(\hat{\bm h}^{\rm b}_{n}\). These two losses jointly enhance the model’s robustness and effectiveness in handling user interest shifts. During testing, the model leverages these self-supervised losses to perform gradient descent, adapting to the input data and improving prediction performance.}
\label{main_model2}
\end{figure*}
\section{\ourname: The Proposed Method}
To track shifts in user interest during the testing phase of sequential recommendation tasks, we propose \ourname, a method that integrates a state space model for test-time training and introduces two alignment modules to adaptively capture these shifts during testing. First, we introduce the base model of the proposed method, which is developed based on the Mamba-2 architecture and can dynamically represent user interest with high test-time throughput. Next, we describe the two alignment modules, designed for time intervals and user interest states, respectively. Finally, we explain how to conduct self-supervised training during the testing phase.

\subsection{Base Model}\label{T3:overall}
As illustrated in Figure~\ref{main_model1}, we develop the basic framework of \ourname\ by stacking the embedding layer, \ourblockname\ block, and the prediction layer. 

\subsubsection{Embedding Layer}\label{emb}
Given a learnable embedding layer \( \bm{E} = [\bm{e}_1,\bm{e}_2,\ldots, \bm{e}_{|\mathcal{V}|}]^\top \in \mathbb{R}^{|\mathcal{V}| \times d} \) for all items, where \( d \) is the embedding dimension and \( \mathbf{e}_j \) represents the dense vector for item \( v_j \), this layer transforms the sparse item id sequence \( \mathcal{S} = [ v_1, v_2, \dots, v_{n} ] \) into dense vector representations, denoted as $\bm E^{(1)}_{\mathcal{S}}\in \mathbb{R}^{n\times d}$.

\subsubsection{T$^3$A-Mamba Block}\label{t3amamba}
Though Transformer-based models excel in sequential recommendation~\cite{kang2018self}, their quadratic complexity related to sequence length impedes efficient test-time training and burdens throughput. To address this issue, we develop \ourblockname\ based on Mamba-2 \cite{dao2024transformers}. 

The output representations $\bm E_{\mathcal{S}}$ of the embedding layer then enter our core \ourblockname\ Block. The sequence undergoes a series of transformations to generate inputs for \ourssmname\ block. Specifically, $\bm E_{\mathcal{S}}$ first passes through a linear layer:
\begin{equation}
    \bm {E}^{(2)}_{\mathcal{S}}, \bm {E}^{(3)}_{\mathcal{S}} \leftarrow \text{Linear}_{1}(\bm E^{(1)}_{\mathcal{S}}), \label{linear}
\end{equation}
where $\bm E_{\mathcal{S}}^{(2)} \in \mathbb{R}^{n\times (d+2\times d_{\rm s})}$, $\bm {E}^{(3)}_{\mathcal{S}} \in \mathbb{R}^{n}$ and Linear$_1$ is a linear parameterized projection from dimension $d$ to dimension $(d+2\times d_s+1)$.
Next, $\bm E_{\mathcal{S}}^{(2)}$ passes through a convolution and a non-linear activation, $\bm {E}^{(3)}_{\mathcal{S}}$ passes through a softplus function:
\begin{equation}
    \bm{X}, \bm B, \bm C \leftarrow \sigma\left(\text{Conv}\left(\bm E_{\mathcal{S}}^{(2)}\right)\right), \quad \bm \Delta=\tau\left(\bm {E}^{(3)}_{\mathcal{S}}\right)\label{conv:TTAA}
\end{equation}
where $\bm X  = [\bm x_1, \bm x_2, \ldots, \bm x_n]^\top\in \mathbb{R}^{n\times d}$ is the input of \ourssmname, $\bm B \in \mathbb{R}^{n\times d_{\rm s}}$ and $\bm C \in \mathbb{R}^{n\times d_{\rm s}}$ are the adjustable matrices, $\bm \Delta \in \mathbb{R}^{n}_{+}$ is the step size ensured to be positive via the softplus function $\tau$, Conv is the convolution operation, $\sigma$ is the non-linear activation. 

The overall transformation process of Equation~\eqref{linear} and \eqref{conv:TTAA} can be unified into a sequence of mappings from the input embeddings $\bm E_{\mathcal{S}}$ to the outputs $\bm X$, ${\bm B}$, $\bm C$, and the step size $\bm \Delta$, denotes as:
\begin{equation}
\label{transform}
    \bm X, {\bm B}, \bm C,\bm \Delta \leftarrow \text{Transform}(\bm E_{\mathcal{S}}^{(1)}).
\end{equation}

Then, we introduce the core ingredient of \ourblockname, the \ourssmname.
Given a learnable scalar value \( A <0 \), we compute the discretized \( \bar{\bm{A}} =[\bar{A}_1, \bar{A}_2, \ldots, \bar{A}_n]^\top \in \mathbb{R}^{n}\) and \( \bar{\bm{B}} =[\bar{\bm B}_1, \bar{\bm B}_2, \ldots, \bar{\bm B}_n]^\top\in \mathbb{R}^{n\times d_{\rm s}}\) follows the Zero-Order Hold (ZOH)~\cite{gu2021combining} method, formally:
\begin{equation}
    \bar{\bm A}=e^{\bm\Delta A},\quad\bar{\bm B}=\text{diag}(\bm\Delta) \bm B,
    \label{discretize}
\end{equation}
simplified as $\bar{\bm A}, \bar{\bm B} \leftarrow \text{discretize}(\bm\Delta, A, \bm B)$, where $\text{diag}(\bm \Delta)$ is a \( n \times n \) diagonal matrix, the diagonal elements are the elements of \(\bm \Delta\).

Finally, given an all-zero matrix \(\bm{h}_0 \in \mathbb{R}^{d_{\rm s}\times d}\) as the initial state, we can iteratively compute the outputs and hidden states in \ourssmname:
\begin{subequations}  
\begin{align}
\bm h_t &=\bar{A}_t \bm h_{t-1}+\bar{\bm B}_t \otimes\bm x_t, \label{dssm21} \\ \bm y_t &=\bm h_t^\top\bm C_t , \label{dssm22}
\end{align}
\end{subequations}
where \(\otimes\) represents the outer product, which combines \(\bar{\bm B}_t \in \mathbb{R}^{d_s}\) and \(\bm x_t \in \mathbb{R}^d\) into a matrix of size \(d_s \times d\), \(\bm Y = [\bm y_1, \bm y_2, \ldots, \bm y_n]^\top\in \mathbb{R}^{n\times d}\),  \(\bm H = [\bm h_1, \bm h_2, \ldots, \bm h_n]^\top\in \mathbb{R}^{nd\times d_{\rm s}}\) are the outputs and hidden states of \ourssmname\ block, separately, and $\bm h_n \in \mathbb{R}^{d_{\rm s}\times d}$ is referred to as the \emph{final state} that characterizes the user's current preference. 

\subsubsection{Feed Forward Network Layer and Prediction Layer}\label{pred}
$\bm Y$ is then passed through a Feed Forward Network (FFN) layer to adapt the features to the semantic space of the next layer: 
\begin{equation}
    \bm{O} = \text{FFN}(\bm Y),
\end{equation}
where \(\bm O = [\bm o_1, \bm o_2, \ldots, \bm o_n]^\top\in \mathbb{R}^{n\times d}\) is the output embeddings. In addition, this process involves residual connections, layer normalization, and other transformations, which are not explicitly represented in the equations for the sake of simplicity.

Based on the output embeddings \(\bm O \) generated by the FFN. We use the last element $\bm o_n$ to predict the next item the user is likely to interact with. The prediction layer computes logits for all items as:
\begin{equation}
\hat{\bm{z}} = \text{Softmax}\left(\bm{E}\bm{o}_{n}\right),
\end{equation}
where $\hat{\bm{z}}\in \mathbb{R}^{|\mathcal{V}|}$, Softmax is the softmax function and $\bm{E}$ is the embedding table of the embedding layer.
\subsubsection{Recommendation Loss}\label{rec_loss}
The recommendation loss is then computed using the cross-entropy loss:
\begin{equation}
\mathcal{L}_{\text{rec}} = -\frac{1}{n} \sum_{i=1}^{n} z_i \log \left(\hat{z}_i\right),
\end{equation}
where $z_i$ denotes the ground-truth for item $i$, $\hat{z}_i$ denotes the predicted logit for item $i$ in $\hat{z}$. This loss encourages the model to maximize the predicted probability of the true interacted item, improving recommendation accuracy.

\subsection{Time Interval Alignment}\label{sslsec1}
To more accurately capture the temporal information in distribution shifts, we introduce a time interval alignment loss in \ourname. In this process, we not only consider relative time but also incorporate absolute timestamps to better reflect the dynamics of temporal information. As shown in the left half of Figure~\ref{main_model2}, we illustrate the computation method for the time interval alignment loss. In the \ourssmname\ block, the discretization process (Equation~\eqref{discretize}) involves the use of time steps $\bm{\Delta} = [\Delta_1, \Delta_2, \ldots, \Delta_n]^\top \in \mathbb{R}_+^n$, where $\Delta_i \in \mathbb{R}_+$. Ensuring the correctness of the time steps is crucial for alleviating user interest shifts, as accurate time steps can better capture the temporal information embedded in distribution shifts.

As discussed in Section~\ref{t3amamba}, the time steps $\bm{\Delta}$ in the \ourblockname\ block are not fixed but are dynamically generated based on the input sequence $\mathcal{S}$ through the embedding layer output $\bm{E}_{\mathcal{S}}$. This dynamic adjustment enables the model to adapt to variations in input features, resulting in a more flexible temporal representation. However, $\bm{\Delta}$ can only adaptively learn the temporal information in timestamp sequence \([ t_1, t_2, \dots, t_{n} ] \), while ignoring the prediction timestamp $t_{n+1}$. To address this, we aim to obtain an adaptive test-time temporal representation $\Delta_{n+1} \in \mathbb{R}_+$ in \ourssmname.

To achieve this, although we cannot directly access the ground truth item $\bm{v}_{n+1}$ during prediction, we consider the output embedding of the feed forward network layer \(\bm{o}_n \in \mathbb{R}^d\). This embedding is used in the prediction layer to compute similarity scores with all item embeddings via a dot product (as detailed in Section~\ref{pred}). During training, \(\bm{o}_n\) is expected to progressively converge toward the embedding representation of the next ground truth item \(\bm{v}_{n+1}\). Based on this observation, we re-feed \(\bm{o}_n\) into the \ourblockname\ block to estimate the adaptive time step $\Delta_{n+1}$ using Equation~\eqref{transform}:
\begin{equation}
    \label{transform2}
    \hat{\bm x}_{n+1}, {\bm B}_{n+1}, \bm C_{n+1}, \Delta_{n+1} \leftarrow \text{Transform}(\bm o_n).
\end{equation}

After obtaining the adaptive time steps $\bm{\Delta}$ and ${\Delta}_{n+1}$, we demonstrate how to align them with the timestamps $t_1, t_2, \dots, t_{n+1}$. First, we compute the time interval sequence based on these timestamps and pad it with 0 at the beginning (since the timestamps of interactions prior to the input sequence are unavailable):
\begin{equation}
\label{ti}
\bm{T} = [0, t_2-t_1, t_3-t_2, \dots, t_{n+1}-t_n],
\end{equation}
where $\bm{T} \in \mathbb{R}^{n+1}$ serves as the ground truth for $\bm{\Delta}$ and $\Delta_{n+1}$.

Then we propose using a pairwise loss to align $\bm{\Delta}$ and ${\Delta}_{n+1}$ with the ground truth time interval $T$, excluding the padded $0$. Specifically, 
we compute the pairwise self-supervise loss as follows:
\begin{equation}
\mathcal{L}_{\text{pairwise}} = \sum_{2\leq i < j\leq n+1} \max\left\{0, 1 - (\Delta_i - \Delta_j) \cdot \left(\frac{T_i - T_j}{\lambda}\right) \right\},   
\end{equation}
where \((\Delta_i - \Delta_j)\) represents the predicted pairwise differences, \((T_i - T_j)\) represents the ground truth pairwise differences, $\lambda$ is constant value for scaling. The loss penalizes mismatches between the relative signs of predicted and ground truth time differences, preserving the relative relationships between the learned time intervals instead of relying on their absolute values.

However, directly applying this loss introduces certain challenges. When handling long sequences (e.g., $50$–$200$ items), the computational complexity of pairwise difference calculations increases to \(O(n^2)\). To address this issue, we adopt a block-based computation approach to simplify the process. Specifically, the sequence is divided into smaller, non-overlapping blocks of size \(b\), and pairwise losses are computed independently within each block. For each block, we first calculate the pairwise differences for predicted values \((\Delta_i - \Delta_j)\) and ground truth values \((T_i - T_j)\). Next, we apply a mask to exclude invalid pairs, accommodating sequences of varying lengths. Finally, the \emph{time interval alignment loss} for each block is computed using the hinge loss formulation and normalized by the total number of valid pairs across all blocks:
\begin{equation}
\mathcal{L_{\mathrm{time}}} = \frac{\sum_{\text{block}} \sum_{2\leq i < j\leq n+1} \max\left\{0, 1 - (\Delta_i - \Delta_j) \cdot \left(\frac{T_i - T_j}{\lambda}\right) \right\}}{\text{Total Valid Pairs}}.
\label{ssl1}
\end{equation}  
This approach reduces both computational and memory complexity from \(O(n^2)\) to \(O(b^2 \times \lceil n/b \rceil)\), significantly enhancing scalability for long sequences.

\subsection{User Interest State Alignment}\label{sslsec2}
To ensure the model accurately captures and represents the evolving patterns of user interests over time, we introduce a \oursslname\ alignment loss in \ourname. As shown in Figure~\ref{intro}, user interest shifts during testing often occur toward the end of the input sequence. Therefore, it is crucial for recommendation models to better understand and align with the user’s interests at the tail end. To achieve this, we align the final states of \ourblockname.

As illustrated on the right side of Figure~\ref{main_model2}, we present the computation process for the \oursslname\ alignment loss. First, we introduce a backward state update function, which is then applied to align the final state of the input user history sequence. Ensuring the correctness of the final states generated by the model is critical for alleviating shifts in user interests.

\subsubsection{Backward State Update Function}
Given the forward state update function $\bm h'(t) = A \bm h(t) + \bm B \bm x(t)$, similar to Equation~\eqref{ssm1}, multiplying both sides by  $A^{-1}$ yields its backward form:
\begin{equation}
\begin{aligned}
\bm h (t) &= A^{-1} \bm h'(t) + \left( - A^{-1} \bm B \bm x(t) \right), 
\end{aligned}
\label{pq:multi:both}
\end{equation}
where $\bm h'(t)$ denotes the derivative of 
 $\bm h(t)$. 
To simplify the notation, we define:
\begin{equation}
\begin{aligned}
P^{\rm b} = A^{-1}, \quad \bm Q^{\rm b} = -A^{-1} \bm B.
\label{pq}
\end{aligned}
\end{equation}
Substituting Equation~\eqref{pq} into Equation~\eqref{pq:multi:both} yields
\begin{equation}
\label{dispq:bsuf}
\bm h (t) = P^{\rm b} \bm h'(t) + \bm Q^{\rm b} \bm x(t).
\end{equation}
For the discrete-time system, we adopt the Zero-Order Hold~\cite{gu2021combining} method for the discretize process with a scalar step size $\Delta$ as follows:
\begin{equation}
\label{dispq}
\bar{P}^{\rm b} = e^{\Delta P^{\rm b}}, \qquad \bar{\bm Q}^{\rm b} = \Delta \bm Q^{\rm b}.
\end{equation}
This discretize process can be simplified as $\bar{P}, \bar{\bm Q} = \text{discretize}(\Delta, P, \bm Q)$. 
Using this parameterization, Equation~\eqref{dispq:bsuf} can be expressed as a \emph{backward discrete state update function} as follows:
\begin{equation}
\begin{aligned}
\bm h^{\rm b}_{t} &= \bar{P}^{\rm b} \bm h_{t+1} + \bar{\bm Q}^{\rm b} \bm x_{t}.
\end{aligned}
\label{bds}
\end{equation}


\subsubsection{Interest State Alignment Loss}
In the \ourssmname\ block explained in Section~\ref{t3amamba}, as described in Equation~(\ref{dssm21}) and (\ref{dssm22}), \ourssmname\ block generates the hidden states \(\bm H = [\bm h_1, \bm h_2, \ldots, \bm h_n]^\top\). We take its final state \( \bm h_n \in \mathbb{R}^{d_{\rm s}\times d}\) and subsequently perform the \oursslname\ alignment operation on it. This is because $\bm h_n$ encapsulates the entirety of the user's historical information while retaining the most recent behavioral context with maximal fidelity.

To begin, we estimate the forward state for the next step at $t_{n+1}$~(abbreviated as \emph{forward state}), denoted as \(\hat{\bm{h}}_{n+1}\in\mathbb{R}^{d_{\rm s}\times d}\), after \(\bm{h}_n\). Following the approach in Equation~\eqref{transform2}, the feed-forward network layer output \(\bm{o}_n\) is transformed to derive \(\hat{\bm{x}}_{n+1}\in\mathbb{R}^{d}\), \({\bm{B}}_{n+1}\in\mathbb{R}^{d_{\rm s}}\), \(\bm{C}_{n+1}\in\mathbb{R}^{d_{\rm s}}\), and \(\Delta_{n+1}\in\mathbb{R}_+\). These terms are subsequently used to compute \(\hat{\bm{h}}_{n+1}\).
Specifically, the discretized terms \(\bar{A}_{n+1} \in \mathbb{R}\) and \(\bar{\bm{B}}_{n+1} \in \mathbb{R}^{d_{\rm s} \times d}\) are computed as follows:
\begin{equation}
    \bar{A}_{n+1}=e^{\Delta_{n+1} A},\quad\bar{\bm B}_{n+1}=\Delta_{n+1} \bm B_{n+1},
    \label{discretize}
\end{equation}
Using these, along with \(\bm{h}_n\) and \(\hat{\bm{x}}_{n+1}\), we update the forward state:
\begin{equation}
\hat{\bm h}_{n+1} = \bar{A}_{n+1} \bm h_{n} + \bar{\bm B}_{n+1}\otimes \hat{\bm x}_{n+1}.
\label{forward}
\end{equation}

Next, we introduce a \emph{backward state}, denoted as \(\hat{\bm h}_{n}^{\rm b}\), and evaluate its alignment with the original state \(\bm h_{n}\) in Equation~\eqref{dssm21} using an additional loss termed \emph{interest state alignment loss}. In simple terms, this involves substituting \(\hat{\bm h}_{n+1}\) from Equation~\eqref{forward} into the backward discrete state update function in Equation~\eqref{bds}, estimate \(\bar{P}^{\mathrm{b}}\) and \(\bar{\bm Q}^{\mathrm{b}}\) using neural network,  
for computing the backward state. 
More specifically, the process begins by taking \(\bm{x}_n \in \mathbb{R}^d\) in Equation~\eqref{conv:TTAA} as input and passing it through a linear layer to produce the estimated matrix \(\bm Q\in \mathbb{R}^{d_{\rm s}}\): 
\begin{equation}
    \bm{Q} \leftarrow \text{Linear}_{2}(\bm{x}_n).\\
\end{equation}
where Linear$_2$ is a linear projection from dimension $d$ to dimension $d_{\rm s}$. Then, for estimating \(\bar{P}^{\mathrm{b}}\) and \(\bar{\bm Q}^{\mathrm{b}}\) in  Equation~(\ref{dispq}) with $\Delta_{n+1}$, we introduce a scalar value $P$, as specified in Theorem~\ref{TTTA:theo1}, and compute the discretized \(\bar{P}\) and \(\bar{\bm Q}\):
\begin{equation}
\label{eq:TTA:P_Q_bar}
    \bar{P} = e^{\Delta_{n+1} P}, \qquad \bar{\bm Q} = \Delta_{n+1} \bm Q.
\end{equation}
Subsequently, substituting Equation~\eqref{eq:TTA:P_Q_bar} into the backward discrete state update function in Equation~\eqref{bds}, we derive the following \emph{backward state}: 
\begin{equation}
\hat{\bm h}^{\rm b}_{n} = \bar{P} \hat{\bm h}_{n+1} + \bar{\bm Q}\otimes{\bm x}_{n}.
\label{backward}
\end{equation}
Finally, we define the \emph{interest state alignment loss} as follows:
\begin{equation}
\mathcal{L}_{\mathrm{state}} = \frac{\left\|\bm h_{n} - \hat{\bm h}^{\rm b}_{n}\right\|_2}{\Delta_{n+1}^2},
\label{ssl2}
\end{equation}
where the coefficient \(\frac{1}{\Delta_{n+1}^2}\) serves as a dilution term as analyzed in  Theorem~\ref{TTTA:theo1}. Intuitively, the interest state alignment loss ensures that the model's internal representation of the user's current interest aligns with the expected state derived from the sequence. Specifically, it reconstructs the user's interest state at the end of their interaction sequence and aligns it with a backwardly updated version of the same state, enabling the model to fine-tune its understanding during the testing phase. 
Theoretically, we provide an upper bound of  $\mathcal{L}_{\mathrm{state}}$ to analyze its effect.
\begin{theorem}
    \label{TTTA:theo1}
Denote $\bm  \epsilon_n = \bm Q - \bm Q^{\mathrm{b}}_{n+1}$ and $\bm Q^{\mathrm{b}}_{n+1} = - A^{-1}\bm{B}_{n+1} $, and let $P = \frac{\ln\left(- A^{-1}\right)}{\Delta_{n+1}}$, 
the following upper bound for $\mathcal{L}_{\mathrm{state}}$ holds:
\begin{equation*}
\begin{aligned}
\mathcal{L}_{\mathrm{state}} \leq \Delta_{n+1}^{-2}\left\| \bm{h}_n \right\|_2
+ \Delta_{n+1}^{-1}\left\|\bm{x}_n\right\|_2 \left\| \bm \epsilon_n \right\|_2
+ A^{-1}\left\|\bm{B}_{n+1}\right\|_2 \left\|\frac{\bm{x}_n - \hat{\bm{x}}_{n+1} }{\Delta_{n+1}}\right\|_2. \\
\end{aligned}
\end{equation*}
\end{theorem}
In Theorem~\ref{TTTA:theo1}, \(\|\epsilon_n\|_2\) is considered an instance-dependent term that depends on the input \(\bm{x}_n\), and it tends to increase for unseen user interaction \(\bm{x}_n\). The term \( \left\|(\bm{x}_n - \hat{\bm{x}}_{n+1}) / \Delta_{n+1} \right\|_2  \) aims to ensure consistency between the model and the most recent interaction under dynamic behavior changes, especially when the time interval \(\Delta_{n+1}\) between the user's testing time and their most recent behavior is very short. Similar to how humans continuously update expectations based on new experiences, this alignment mechanism enables the model to dynamically adapt to changes in user behavior, thereby improving recommendation accuracy.

\begin{algorithm}[tbp]
\caption{Test-Time Alignment for \ourname~ on a Single Batch}
\label{alg:test_time_training}
\raggedright
\textbf{Input:} Batch of test sequences, denoted as $\{\mathcal{S}_i\}_{i=1}^{m}$, where each represented as \(\mathcal{S}_i = [v_1, v_2, \dots, v_n]\) with the corresponding timestamp sequence \([ t_1, t_2, \dots, t_{n} ] \) and timestamp of prediction \(t_{n+1}\), well-trained model \(g_{\bm \theta}(\cdot)\), number of training steps \(M\), learning rate \(\alpha\), and weight parameters \(\mu_1^{\mathrm{test}}\) and \(\mu_2^{\mathrm{test}}\)

\textbf{Output:} Final prediction \(\bm{o}_n\)

\begin{algorithmic}[1]
\STATE Record the original model parameters: \(\bm \theta_{o} \gets \bm \theta\)
\FOR{\(\text{step} = 1, 2, \ldots, M\)}
    \STATE Forward \(g_{\bm \theta}(\{\mathcal{S}_i\}_{i=1}^{m})\) and get the learned step size $\Delta$, final state \(\bm h_n\) and output \(\bm{o}_n\)
    \STATE Compute the time interval sequence $\bm T$ using Equation~(\ref{ti})
    \STATE Calculate \(\mathcal{L}_{\mathrm{time}}\) from Equation~(\ref{ssl1}) using \(\bm\Delta\) and $\bm T$
    \STATE Compute the forward state \(\hat{\bm h}_{n+1}\) using Equation~(\ref{forward}) 
    \STATE Compute the backward state \(\hat{\bm h}^{\rm b}_{n}\) using Equation~(\ref{backward})
    \STATE Calculate \(\mathcal{L}_{\mathrm{state}}\) from Equation~(\ref{ssl2}) using \(\bm h_{n}\) and \(\hat{\bm h}^{\rm b}_{n}\)
    \STATE Update model parameters:
    \[
    {\bm \theta} \gets {\bm \theta} - \alpha \nabla_{\bm \theta} \left(\mu_1^{\mathrm{test}}\mathcal{L}_{\mathrm{time}} + \mu_2^{\mathrm{test}}\mathcal{L}_{\mathrm{state}}\right)
    \]
\ENDFOR
\STATE Compute final prediction \(\bm{o}_n\) using the updated \(g_{\bm \theta}(\{\mathcal{S}_i\}_{i=1}^{m})\)
\STATE Restore the model parameters: \({\bm \theta} \gets {\bm \theta}_{o}\) for the next batch
\end{algorithmic}
\end{algorithm}

\begin{table}[t]
\centering
\caption{Dataset statistics.}
\label{sta}
\begin{tabular}{l|r|r|r}
\hline
\textbf{Dataset} & \textbf{ML-1M} & \textbf{Amazon} & \textbf{Zhihu-1M} \\ 
\hline
\#Users & 6,034 & 247,659 & 7,974\\
\#Items & 3,706 & 10,814 & 81,563\\
\#Interactions & 1,000,209 & 471,615 & 999,970\\
Avg. Length & 138.3 & 17.04 & 29.03 \\
Sparsity & 95.57\% & 98.78\% & 99.48\% \\
\hline
\end{tabular}
\end{table}

\begin{table*}[t]
\centering
  \caption{Performance comparison of different sequential recommendation models. The best result is bolded and the runner-up is underlined. * means improvements over the second-best methods are significant (\textit{t}-test, \textit{p}-value < 0.05).}
  \label{tab:main}
\begin{tabular}{l|ccc|ccc|ccc}
\hline
\multirow{2}{*}{Model} & \multicolumn{3}{c}{ML-1M} & \multicolumn{3}{c}{Amazon} & \multicolumn{3}{c}{Zhihu-1M} \\
\cline{2-4} \cline{5-7} \cline{8-10}
& Recall@10 & MRR@10 & NDCG@10 & Recall@10 & MRR@10 & NDCG@10 & Recall@10 & MRR@10 & NDCG@10  \\
\hline
\hline
Caser & $0.1954$ & $0.0703$ & $0.0994$ & $0.0594$ & $0.0204$ & $0.0294$ & $0.0288$ & $0.0089$ & $0.0134$\\
GRU4Rec & $0.2732$ & $0.1147$ & $0.1518$ & $0.0939$ & $0.0480$ & $0.0586$ & $0.0283$ & $0.0092$ & $0.0142$\\
BERT4Rec & $0.2770$ & $0.1093$ & $0.1482$ & $0.0675$ & $0.0372$ & $0.0443$ & $0.0289$ & $0.0098$ & $0.0142$\\
SASRec & $0.2471$ & $0.0911$ & $0.1273$ & $\underline{0.1025}$ & $0.0527$ & $0.0644$ & $0.0364$ & $0.0098$ & $0.0159$ \\
Mamba4Rec & $0.2813$ & $0.1201$ & $0.1578$ & $0.1003$ & $0.0522$ & $0.0635$ & $0.0355$ & $0.0112$ & $0.0167$\\
TiM4Rec & $0.2873$ & $\underline{0.1211}$ & $0.1596$ & $0.1016$ & $\underline{0.0567}$ & $\underline{0.0665}$ & \underline{$0.0384$} & \underline{$0.0118$} & \underline{$0.0179$}\\
TTT4Rec& $\underline{0.2887}$ & ${0.1208}$ & $\underline{0.1599}$ & $0.1020$ & $0.0560$ & $0.0655$ & $0.0370$ & $0.0115$ & $0.0174$\\
\ourname ~(ours) & $\mathbf{0.2932}^{*}$ & $\mathbf{0.1262}^{*}$ &$\mathbf{0.1648}^{*}$ & $\mathbf{0.1102}^{*}$ & $\mathbf{0.0580}^{*}$ & $\mathbf{0.0705}^{*}$ & $\mathbf{0.0402}^{*}$ & $\mathbf{0.0122}^{*}$ & $\mathbf{0.0187}^{*}$\\
\hline
\end{tabular}
\end{table*}

\subsection{Training and Testing Process}
\subsubsection{Training Process}\label{training}
After designing the two self-supervised losses, the total loss of the model combines the primary recommendation loss \(\mathcal{L}_{\text{rec}}\) in Section~\ref{rec_loss} and the two self-supervised losses \(\mathcal{L}_{\mathrm{time}}\) and \(\mathcal{L}_{\mathrm{state}}\), weighted by their respective hyperparameters:
\begin{equation}
\mathcal{L}_{\text{total}} = \mathcal{L}_{\text{rec}} + \mu^{\mathrm{train}}_1 \mathcal{L}_{\mathrm{time}} + \mu^{\mathrm{train}}_2 \mathcal{L}_{\mathrm{state}},
\label{main_loss}
\end{equation}
where \(\mu^{\mathrm{train}}_1\) and \(\mu^{\mathrm{train}}_2\) are the weights for the self-supervised losses. This combined loss optimizes the model by integrating both the primary task and self-supervised learning objectives.

\subsubsection{Testing Process}
During testing, the pre-trained model $g_{\theta}(\cdot)$'s parameters $\theta$ are adaptively fine-tuned on all testing examples. For each batch of test data $\{\mathcal{S}_i\}_{i=1}^{m}$, optimization is performed using the self-supervised losses \(\mathcal{L}_{\mathrm{time}}\) and \(\mathcal{L}_{\mathrm{state}}\) as defined in Equation~(\ref{ssl1}) and Equation~(\ref{ssl2}). The trained model is then used to make the final item predictions. After completing the predictions, the adjusted parameters on $\theta$ are discarded, and the model reverts to its original parameters to process the next batch of testing examples. The pseudo-code for processing a single batch of testing examples is illustrated in Algorithm~\ref{alg:test_time_training}.
\section{Experiments}
To verify the effectiveness of \ourname, we conduct extensive experiments and report detailed analysis results.

\subsection{Experimental Settings}
\subsubsection{Datasets}
We evaluate the performance of the proposed model through experiments conducted on three public datasets:
\begin{itemize}[leftmargin=*]
    \item MovieLens-1M (referred to as \textbf{ML-1M})~\cite{harper2015movielens}: A dataset collected from the MovieLens platform, containing approximately 1 million user ratings of movies.
    \item Amazon Prime Pantry (referred to as \textbf{Amazon})~\cite{ni2019justifying}: A dataset of user reviews in the grocery category collected from the Amazon platform up to 2018.
    \item \textbf{Zhihu-1M}~\cite{hao2021large}: A dataset sourced from a large knowledge-sharing platform (Zhihu), consisting of raw data, including information on questions, answers, and user profiles.
\end{itemize}
For each user, we sort their interaction records by timestamp to generate an interaction sequence. 
We retain only users and items associated with at least ten interaction records.
We follow the leave-one-out policy~\cite{kang2018self} for training-validation-testing partition. The statistical details of these datasets are presented in Table~\ref{sta}.

\subsubsection{Baselines}
To demonstrate the effectiveness of our proposed method, we conduct comparisons with several representative sequential recommendation baseline models:
\begin{itemize}[leftmargin=*]
    \item \textbf{Caser}~\cite{tang2018personalized}: A foundational sequential recommendation model leveraging Convolutional Neural Networks (CNN).
    \item \textbf{GRU4Rec}~\cite{hidasi2015session}: A method based on Recurrent Neural Networks (RNN), specifically utilizing Gated Recurrent Units (GRU) for sequential recommendation.
    \item \textbf{BERT4Rec}~\cite{sun2019bert4rec}: A model that adopts the bidirectional attention mechanism inspired by the BERT~\cite{devlin2018bert} framework.
    \item \textbf{SASRec}~\cite{kang2018self}: The first model to introduce the Transformer architecture to the field of sequential recommendation.
    \item \textbf{Mamba4Rec}~\cite{liu2024mamba4rec}: An innovative model that applies the Mamba architecture to the sequential recommendation domain.
    \item \textbf{TiM4Rec}~\cite{fan2024tim4rec}: A time-aware sequential recommendation model that improves SSD's low-dimensional performance while maintaining computational efficiency.
    \item \textbf{TTT4Rec}~\cite{yang2024ttt4rec}: A framework that uses Test-Time Training to dynamically adapt model parameters during inference for sequential recommendation.
\end{itemize}
\subsubsection{Evaluation Metrics}
To evaluate the performance of top-$K$ recommendation, we adapt the metrics Recall@$K$, MRR@$K$, and NDCG@$K$, which are widely used in recommendation research to evaluate model performance~\cite{qu2024scalable, zhang2024qagcf, qu2024budgeted}. In this context, we set $K = 10$ and present the average scores on the test dataset.
\subsubsection{Implementation Details}
Our evaluation is based on implementations using PyTorch~\cite{paszke2019pytorch} and RecBole~\cite{zhao2021recbole}. We set all model dimensions $d$ to $64$, the learning rate to $0.001$, and the batch size to $4096$. Additionally, we set the state dimension $d_{\mathrm{s}}$ of the \ourblockname\ to $32$, the number of blocks to $1$, the training steps M during testing $M$ to $1$ and learning rate during testing $\alpha$ to $0.005$. Furthermore, the search space for the weights of the two self-supervised losses during training, $\lambda^{\mathrm{train}}_1$ and $\lambda^{\mathrm{train}}_2$, is $\{0.01, 0.1, 1, 10\}$, and for testing, the search space for $\lambda^{\mathrm{test}}_1$ and $\lambda^{\mathrm{test}}_2$ is $\{ 1e\text{-}3, 1e\text{-}2, 1e\text{-}1, 1\}$. To adapt to the characteristics of the baselines and datasets, we set the fixed sequence length to 200 for ML-1M and 50 for other datasets (Amazon and Zhihu-1M). For a fair comparison, we ensured consistency of key hyperparameters across different models while using default hyperparameters for baseline methods as recommended in their respective papers. Finally, we utilized the Adam optimizer~\cite{kingma2014adam} in a mini-batch training manner.

\subsection{Overall Performance}
Table~\ref{tab:main} presents the performance comparison of the proposed \ourname\ and other baseline methods on three datasets. From the table, we observe several key insights:
\begin{enumerate}[leftmargin=*]
    \item General sequential methods, such as SASRec and Mamba4Rec, show varying strengths depending on the dataset characteristics. Mamba4Rec performs better on ML-1M, a dataset with longer sequences, due to its ability to model complex long-term dependencies, whereas SASRec achieves slightly higher performance on shorter sequence datasets like Amazon and Zhihu-1M, where simpler sequence modeling suffices.
    \item TiM4Rec and TTT4Rec outperform these general sequential methods by alleviating specific challenges in recommendation. Tim4Rec effectively incorporates temporal information, enhancing its ability to model time-sensitive dynamics, while TTT4Rec leverages test-time training to alleviate distribution shifts, leading to superior performance.
    \item Finally, \ourname\ outperforms all baselines across datasets, achieving the best results by introducing time interval alignment loss and \oursslname\ alignment loss. These innovations enable \ourname\ to capture temporal and sequential patterns more effectively and adapt dynamically to test data through gradient descent during inference, further alleviating distributional shifts and enhancing overall performance.
\end{enumerate}
\begin{table*}[t]
\centering
  \caption{Ablation studies.}
  \label{tab:ablation}
\vspace{-0.3cm}
\begin{tabular}{l|ccc|ccc|ccc}
\hline
\multirow{2}{*}{Model} & \multicolumn{3}{c}{ML-1M} & \multicolumn{3}{c}{Amazon} & \multicolumn{3}{c}{Zhihu-1M} \\
\cline{2-4} \cline{5-7} \cline{8-10}
& Recall@10 & MRR@10 & NDCG@10 & Recall@10 & MRR@10 & NDCG@10 & Recall@10 & MRR@10 & NDCG@10  \\
\hline
\hline
\ourname & $0.2932$ & $0.1262$ & $0.1648$ & $0.1102$ & $0.0580$ & $0.0705$ & $0.0402$ & $0.0122$ & $0.0187$\\
\hdashline
w/o $\mathcal{L}_\text{both}$ & $0.2813$ & $0.1207$ & $0.1578$ & $0.1003$ & $0.0522$ & $0.0635$ & $0.0355$ & $0.0112$ & $0.0167$\\
w/o $\mathcal{L}_\mathrm{time}$ & $0.2891$ & $0.1225$ & $0.1618$ & $0.1026$ & $0.0531$ & $0.0649$ & $0.0389$ & $0.0110$ & $0.0180$\\
w/o $\mathcal{L}_\mathrm{state}$ & $0.2876$ & $0.1214$ & $0.1590$ & $0.1010$ & $0.0545$ & $0.0652$ & $0.0365$ & $0.0114$ & $0.0175$\\
\hdashline
w/o $\mathcal{L}_\text{both}^{\mathrm{test}}$ & $0.2899$ & $0.1239$ & $0.1626$ & $0.1033$ & $0.0568$ & $0.0677$ & $0.0389$ & $0.0113$ & $0.0177$\\
w/o $\mathcal{L}_\mathrm{time}^{\mathrm{test}}$ & $0.2912$ & $0.1252$ & $0.1645$ & $0.1095$ & $0.0570$ & $0.0690$ & $0.0390$ & $0.0120$ & $0.0180$\\
w/o $\mathcal{L}_\mathrm{state}^{\mathrm{test}}$ & $0.2921$ & $0.1241$ & $0.1635$ & $0.1067$ & $0.0558$ & $0.0686$ & $0.0393$ & $0.0120$ & $0.0183$\\
\hline
\end{tabular}
\end{table*}

\subsection{Ablation Studies}
We conduct ablation experiments in \ourname\ to validate the effectiveness of the time interval alignment loss and the \oursslname\ alignment loss. The experimental configurations include `$\mathcal{L}_\text{both}$', where both losses are removed during training and testing, as well as `$\mathcal{L}_\mathrm{time}$' and `$\mathcal{L}_\mathrm{state}$', where only one of the losses is removed. During testing, we design configurations such as `$\mathcal{L}_\text{both}^{\mathrm{test}}$', `$\mathcal{L}_\mathrm{time}^{\mathrm{test}}$', and `$\mathcal{L}_\mathrm{state}^{\mathrm{test}}$', which disable all or part of the losses during the testing phase. As shown in Table~\ref{tab:ablation}, removing any loss significantly degrades model performance, highlighting the importance of adapting to use interest shifts. The time interval alignment loss captures temporal dependencies, while the \oursslname\ alignment loss captures the current users' interest state. Both are crucial for the robustness and adaptability of sequential recommendations.
\begin{figure}[t]
\centering
\includegraphics[width=0.45 \textwidth]{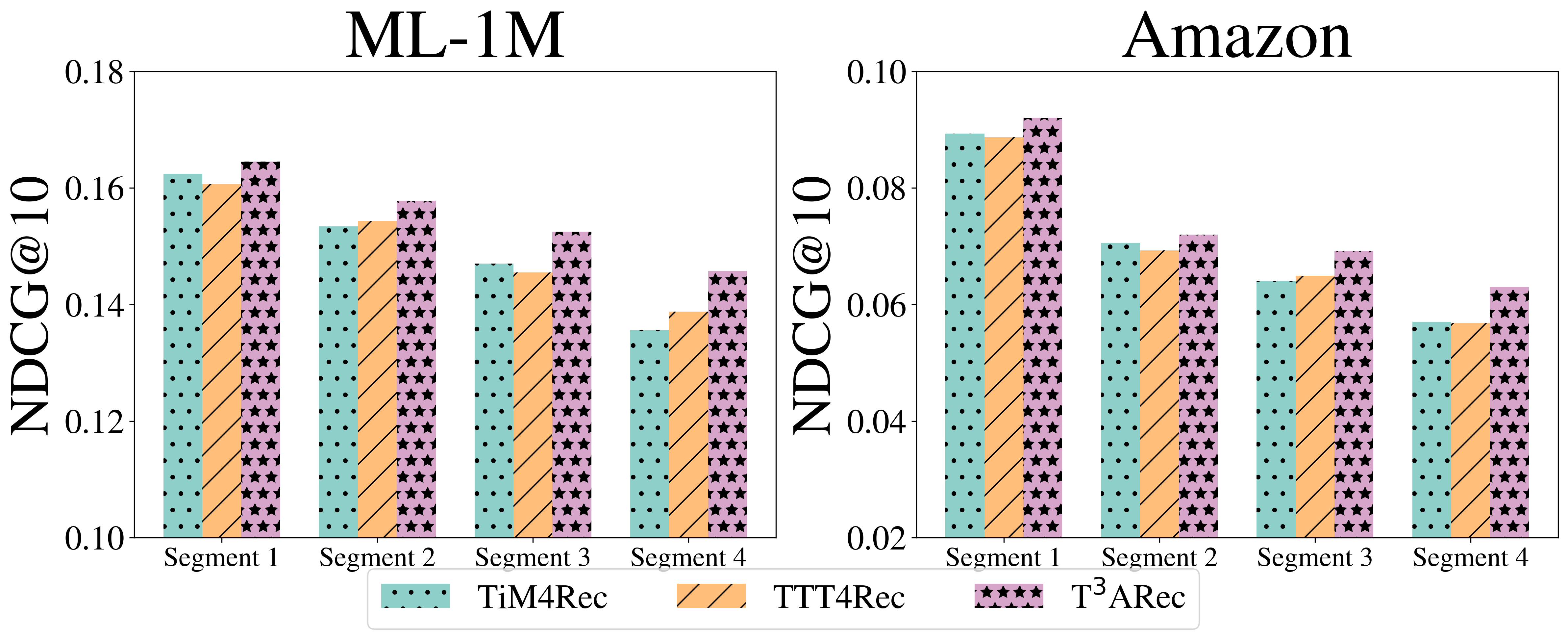}
\vspace{-0.3cm}
\caption{Effectiveness of \ourname\ on user interest shifts.}
\label{eff}
\end{figure}


\begin{table}[t]
\centering
  \caption{Analysis of test-time throughput, defined as the number of iterations per second, with each iteration processing a batch of 4096 testing examples.}
  \label{tab:timeana}
\vspace{-0.3cm}
\begin{tabular}{r|c|c|c}
\hline
\multirow{2}{*}{Model} &  \multicolumn{3}{c}{Test-Time Throughput (\# of iterations~/~second)} \\
\cline{2-4}
& \multicolumn{1}{c}{ML-1M} & \multicolumn{1}{c}{Amazon} & \multicolumn{1}{c}{Zhihu-1M}
\\
\hline
\hline
SASRec & $1.56$ & $2.26$ & $1.96$  \\
Mamba4Rec & $2.82$ & $3.11$ & $3.01$ \\
TiM4Rec & $2.16$ & $2.64$ & $2.56$ \\
TTT4Rec & $0.96$ & $1.74$ & $1.19$ \\
\ourname ~(ours)& $1.02$ & $2.05$ &$1.36$ \\
\hline
\end{tabular}
\end{table}

\begin{figure}[t]
\centering
\includegraphics[width=0.9\columnwidth]{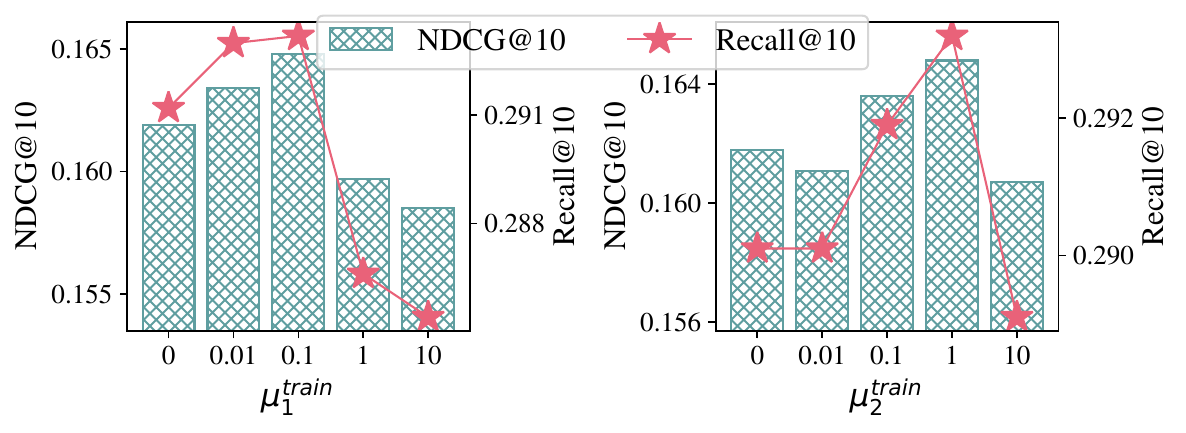}
\vspace{-0.3cm}
\caption{Sensitivity analysis of $\mu_1^{\mathrm{train}}$ and $\mu_2^{\mathrm{train}}$.}
\label{ana1}
\end{figure}

\begin{figure}[t]
\centering
\includegraphics[width=0.9\columnwidth]{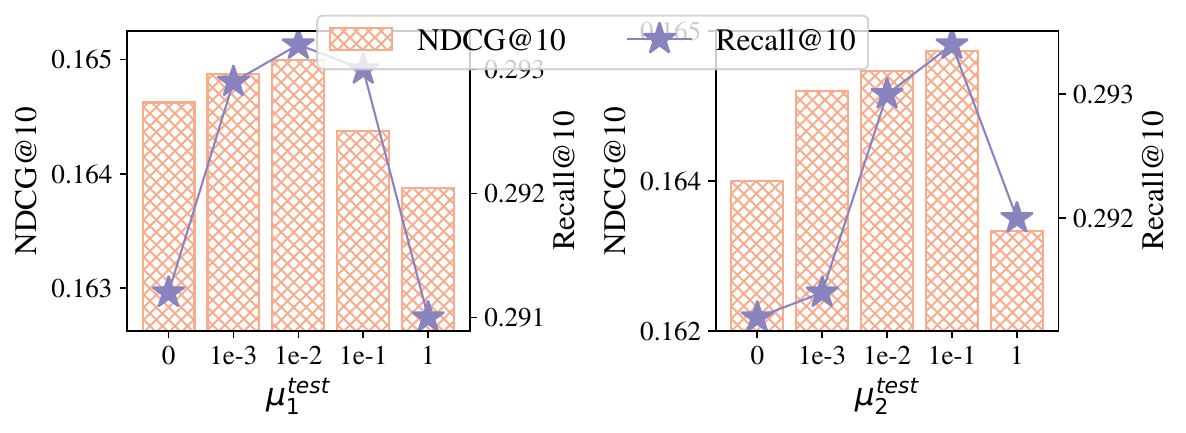}
\vspace{-0.3cm}
\caption{Sensitivity analysis of $\mu_1^{\mathrm{test}}$ and $\mu_2^{\mathrm{test}}$.}
\label{ana2}
\end{figure}

\subsection{Further Analysis}
\subsubsection{Effectiveness of \ourname\ on User Interest Shifts}
To validate the effectiveness of \ourname\ on user interest shifts, we followed the setup in Figure~\ref{vali} and compared \ourname\ with TiM4Rec and TTT4Rec on both the ML-1M and Amazon datasets. As shown in Figure~\ref{eff}, the later the timestamp of the segment, indicating a greater user interest shift, the performance of TiM4Rec and TTT4Rec shows a declining trend. On the basis of these two baselines, our model demonstrates higher improvements in Segment 3 and Segment 4 compared to Segment 1 and Segment 2. This indicates that \ourname\ performs better on test data with more significant user interest shifts, highlighting the effectiveness of our two alignment losses during Test-Time Alignment.


\subsubsection{Analysis of Test-time Throughput}
Noticed that training during testing requires gradient descent, which could introduce additional load on test-time (inference) throughput. We conducted offline experiments to analyze the impact of training during testing. The results are shown in Table~\ref{tab:timeana}, where we compare \ourname\ with several typical baselines and quantify the test-time throughput in terms of it/s (iterations per second). From the results in the table, it is evident that the test-time throughput of the method that trains during testing is approximately half of that of Mamba4Rec that does not train during testing. Nevertheless, with increasing computational power, the additional load from training during testing may not have a significant impact, especially since previous work has shown that TTT has yielded online benefits~\cite{wang2024not}. We present these findings and provide optimization opportunities for future work applying TTT in recommendations.

\subsubsection{Sensitive Analysis of $\mu_1^{\mathrm{train}}$ and $\mu_2^{\mathrm{train}}$}
We analyze the impact of the two self-supervised loss weight parameters, $\mu_1^{\mathrm{train}}$ and $\mu_2^{\mathrm{train}}$, in the training process of Equation~(\ref{main_loss}) in \ourname\ as shown in Figure~\ref{ana1}.
$\mu_1^{\mathrm{train}}$ is the weight of the time interval alignment loss during training, while $\mu_2^{\mathrm{train}}$ is the weight of the \oursslname\ alignment loss during training. The smaller the value of $\mu_1^{\mathrm{train}}$ and $\mu_2^{\mathrm{train}}$, the less influence each loss has on the model training.
From Figure~\ref{ana1}, we observe that the optimal choice for $\mu_1^{\mathrm{train}}$ is $0.1$ and for $\mu_2^{\mathrm{train}}$ is $1$. It is clear that both losses significantly impact the training of \ourname, with the \oursslname\ alignment loss having a slightly greater effect. This also confirms the importance of the model understanding the current user interest pattern.

\subsubsection{Sensitive Analysis of $\mu_1^{\mathrm{test}}$ and $\mu_2^{\mathrm{test}}$}
We analyze the impact of the two self-supervised loss weight parameters, $\mu_1^{\mathrm{test}}$ and $\mu_2^{\mathrm{test}}$, during the testing process of \ourname. 
$\mu_1^{\mathrm{test}}$ is the weight of the time interval alignment loss during testing, while $\mu_2^{\mathrm{test}}$ is the weight of the \oursslname\ alignment loss during testing. The smaller the values of $\mu_1^{\mathrm{test}}$ and $\mu_2^{\mathrm{test}}$, the less impact each loss has on the model during testing.
From Figure~\ref{ana2}, we observe that the optimal choice for $\mu_1^{\mathrm{test}}$ is $1e^{-2}$, and for $\mu_2^{\mathrm{test}}$ it is $1e^{-1}$. It is clear that the \oursslname\ alignment loss has a significant impact on the model's performance during testing. The model needs to correctly learn the current recommendation pattern from the input sequence and adapt accordingly to make better next-time predictions for the user.

\section{Conclusion}
In this work, we alleviate the critical challenge of adapting sequential recommendation to real-world scenarios where user interest shifts dynamically at the test time. Existing methods, reliant on static training data, often fail to adapt effectively to these changes, leading to substantial performance degradation. To overcome this limitation, we introduce \ourname, a novel approach that integrates Test-Time Training (TTT) into sequential recommendation through two alignment-based self-supervised losses. By aligning absolute time intervals with model-adaptive learning intervals and incorporating a state alignment mechanism, \ourname\ captures temporal dynamics and user interest patterns more effectively. Extensive evaluations on multiple benchmark datasets demonstrate that \ourname\ significantly outperforms existing methods, providing robust performance and mitigating the effects of distribution shifts. This work highlights the potential of TTT as a paradigm for enhancing the adaptability of sequential recommendation at the test time.



\appendix
\section{Proof of Theorem \ref{TTTA:theo1}}
\begin{proof}[Proof of Theorem~\ref{TTTA:theo1}]
Based on the definitions of $\mathcal{L}_{\mathrm{state}}$, $\hat{\bm{h}}_n^b$, and $\hat{\bm{h}}_{n+1}$, we can derive: 
\begin{equation*}
\begin{aligned}
&\hphantom{{}={}}
\mathcal{L}_{\mathrm{state}}
\\
&= \left\| \bm{h}_n - \hat{\bm{h}}_n^b \right\|_2 \Delta_{n+1}^{-2} \\
&= \left\| \bm{h}_n - \left( e^{\Delta_{n+1} P} \hat{\bm{h}}_{n+1} + \Delta_{n+1} \bm{Q} \otimes \bm{x}_n \right) \right\|_2\Delta_{n+1}^{-2} \\
&= \left\| \bm{h}_n - \left( e^{\Delta_{n+1} P} \left( e^{\Delta_{n+1} A} \bm{h}_n + \Delta_{n+1} \bm{B}_{n+1} \otimes\hat{\bm{x}}_{n+1} \right) + \Delta_{n+1} \bm{Q} \bm{x}_n \right) \right\|_2\Delta_{n+1}^{-2} \\
&= \left\| \left( 1 - e^{\Delta_{n+1} (P + A)} \right) \bm{h}_n 
+ \Delta_{n+1} \bm{Q} \otimes \bm{x}_n 
-  \Delta_{n+1} e^{\Delta_{n+1} P}\bm{B}_{n+1} \otimes \hat{\bm{x}}_{n+1} \right\|_2\Delta_{n+1}^{-2} \\
&= \Big\| \left( 1 - e^{\Delta_{n+1} \left(P + A\right)} \right) \bm{h}_n 
+ \Delta_{n+1} \left( \bm{Q} - e^{\Delta_{n+1} P}\bm{B}_{n+1}  \right)\otimes \bm{x}_n + \\
&\hphantom{{}=={}} \Delta_{n+1} e^{\Delta_{n+1} P} \bm{B}_{n+1} \otimes \left( \bm{x}_n - \hat{\bm{x}}_{n+1} \right) \Big\|_2\Delta_{n+1}^{-2}, \\
\end{aligned}
\end{equation*}
yielding that
\begin{equation*}
\begin{aligned}
&\hphantom{{}={}}
\mathcal{L}_{\mathrm{state}}
\\
&\leq \Big\| \left( 1 - e^{\Delta_{n+1} (P + A)} \right) \bm{h}_n \Big\|_2\Delta_{n+1}^{-2}
+ \Big\|\Delta_{n+1} \left( \bm{Q} - \bm{B}_{n+1} e^{\Delta_{n+1} P} \right)\otimes \bm{x}_n\Big\|_2\Delta_{n+1}^{-2} + \\
&\hphantom{{}={}}  \Big\|\Delta_{n+1} \bm{B}_{n+1} e^{\Delta_{n+1} P}\otimes \left( \bm{x}_n - \hat{\bm{x}}_{n+1} \right) \Big\|_2\Delta_{n+1}^{-2} ~(\text{triangle inequality})\\
&\leq \Delta_{n+1}^{-2}\left\| \bm{h}_n \right\|_2
+ \Big\|\Delta_{n+1}^{-1} \left( - A^{-1}\bm{B}_{n+1} + \bm  \epsilon_n - e^{\Delta_{n+1} P}\bm{B}_{n+1} \right) \otimes \bm{x}_n\Big\|_2 +\\
&\hphantom{{}={}}  \Big\|\Delta_{n+1}^{-1} e^{\Delta_{n+1} P} \bm{B}_{n+1} \otimes \left( \bm{x}_n - \hat{\bm{x}}_{n+1} \right) \Big\|_2 ~(\text{Def. of~} \bm Q; e^x\leq 1, \forall x\leq0) \\
&= \Delta_{n+1}^{-2}\left\| \bm{h}_n \right\|_2
+ \left\|\Delta_{n+1}^{-1} \left( - \left(A^{-1}+e^{\Delta_{n+1} P}\right)\bm{B}_{n+1} + \bm  \epsilon_n \right)\otimes \bm{x}_n \right\|_2 + \\
&\hphantom{{}={}} \left\|\Delta_{n+1}^{-1} A^{-1} \bm{B}_{n+1}  \otimes \left( \bm{x}_n - \hat{\bm{x}}_{n+1} \right) \right\|_2 \quad (P = \ln (- A^{-1} )/ \Delta_{n+1})\\
&= \Delta_{n+1}^{-2}\left\| \bm{h}_n \right\|_2
+ \left\| \Delta_{n+1}^{-1} \bm \epsilon_n \otimes \bm{x}_n \right\|_2
+ \left\|\Delta_{n+1}^{-1}  A^{-1} \bm{B}_{n+1} \otimes \left( \bm{x}_n - \hat{\bm{x}}_{n+1} \right) \right\|_2 \\
&= \Delta_{n+1}^{-2}\left\| \bm{h}_n \right\|_2
+ \Delta_{n+1}^{-1}\left\|\bm{x}_n\right\|_2 \left\| \bm \epsilon_n \right\|_2
+ A^{-1}\left\|\bm{B}_{n+1}\right\|_2 \left\| \frac{\bm{x}_n - \hat{\bm{x}}_{n+1}}{\Delta_{n+1}} \right\|_2.
\end{aligned}
\end{equation*}
\end{proof}
\bibliographystyle{ACM-Reference-Format}
\bibliography{sample-base}

\end{document}